%% file: paper.tex
\documentclass[a4paper]{article}
\usepackage{spconf,amsmath,amssymb,amsthm,graphicx,cite}
\usepackage[usenames,dvipsnames]{pstricks}
\usepackage{algorithm}
\usepackage[noend]{algpseudocode}

\renewcommand{\figurename}{Fig.}

\title{Overhead-Aware Distributed CSI Selection \\ in the MIMO Interference Channel}
\name{Rami Mochaourab$^*$, Rasmus Brandt$^*$, Hadi Ghauch$^{**}$, Mats Bengtsson$^*$}
\address{$^*$Signal Processing Department, ACCESS Linnaeus Centre, KTH Royal Institute of Technology \\ $^{**}$Communication Theory Department, KTH Royal Institute of Technology}

\input{code}

\begin{document}

\maketitle
\input{sections/abstract}
%
\input{sections/introduction}

\input{sections/system_model}
\input{sections/stable_matching}
\input{sections/simulations}

\bibliographystyle{IEEEbib}
\bibliography{references}

\end{document}

%% file: code.tex
\newtheorem{definition}{Definition}
\newtheorem{theorem}{Theorem}

\newcommand{\mat}[1]{\boldsymbol{#1}}
\newcommand{\tr}[1]{\text{tr}\left( #1\right)}

\newcommand{\br}[1]{{\left\{ #1 \right\}}}

\newcommand{\fnorm}[1]{{ {\Vert #1 \Vert}_{\text{\scriptsize{F}}} }}
\newcommand{\sfnorm}[1]{{ \Vert #1 \Vert^2_{\text{\scriptsize{F}}} }}
\newcommand{\abs}[1]{{ \left| #1 \right| }}

\DeclareMathOperator*{\argmax}{arg\,max}

\def\H{{{\dag}}} 
\def\bH{{\mat{H}}} 
\def\bZ{{\mat{Z}}} 
\def\bI{{\mat{I}}} 
\def\bV{{\mat{V}}} 
\def\bz{{\mat{z}}} 

\def\setT{{\mathcal{T}}} 
\def\setB{{\mathcal{B}}} 
\def\setR{{\mathcal{R}}} 
\def\setM{{\mathcal{M}}} 
\def\setI{{\mathcal{I}}} 
\def\setJ{{\mathcal{J}}} 
\def\setK{{\mathcal{K}}} 
\def\setS{{\mathcal{S}}} 
\def\setP{{\mathcal{P}}} 

%% file: sections/abstract.tex
\begin{abstract}
We consider a MIMO interference channel in which the transmitters and receivers operate in frequency-division duplex mode. In this setting, interference management through coordinated transceiver design necessitates channel state information at the transmitters (CSI-T). The acquisition of CSI-T is done through feedback from the receivers, which entitles a loss in degrees of freedom, due to training and feedback. This loss increases with the amount of CSI-T. In this work, after formulating an overhead model for CSI acquisition at the transmitters, we propose a distributed mechanism to find for each transmitter a subset of the complete CSI, which is used to perform interference management. The mechanism is based on many-to-many stable matching. We prove the existence of a stable matching and exploit an algorithm to reach it. Simulation results show performance improvement compared to full and minimal CSI-T.
\end{abstract}
%

%% file: sections/introduction.tex

\section{Introduction}
\emph{Ultra-dense networks} have been identified as one of the key scenarios in 5G communication systems, characterized by a high number of nodes, located in close proximity, in both heterogeneous and homogeneous networks~\cite{METISD62}. Evidently, \emph{coordination among transmitters and receivers} is vital in such scenarios, where it is well-known that the sum-rate performance is limited by unsuppressed  interference. \emph{Forward-backward training} algorithms, such as~\cite{Gomadam2011,Peters2011, Shi2011}, employ uplink and downlink pilots to estimate the required local channel state information (CSI) quantities, and iteratively refine the precoding / decoding matrices in a cluster of cooperating nodes. 

It becomes clear at this stage that the overhead associated with such clusters is a major concern (especially in the aforementioned dense networks): this motivates the need for schemes that take into account the loss in performance incurred by coordination, training and possibly feedback. This problem was addressed in~\cite{Ayach2012a} where the authors proposed a training and analog feedback scheme that maximizes the \emph{effective sum-rate} (the achievable sum-rate in the network accounting for the loss in degrees of freedom (DoF) due to  overhead). Furthermore, the issue of \emph{user partitioning} was tackled in~\cite{Peters2012} where several schemes for partitioning the users into orthogonal groups are proposed.  

We propose in this work an overhead-aware framework for distributed cooperation in frequency-division-duplex (FDD) systems: sets of potentially cooperating transmitters and receivers are formed using many-to-many stable matching, where the utilities at both transmitters and receivers are designed to take into account both performance of each link, and the associated overhead for channel estimation and feedback. We prove that the formulated model satisfies the conditions for the existence of a many-to-many stable matching. By utilizing an algorithm to reach a stable matching we provide a distributed mechanism to determine the amount of CSI present at each transmitter which is exploited for interference management. This being said, any algorithm for precoder optimization can now be employed to optimize any desired metric (interference leakage, sum-rate, etc.). Finally, our simulations indicate that our proposed distributed scheme offers gains in performance over selected benchmark schemes. 

\emph{Notations:} Column vectors and matrices are given in lowercase and uppercase boldface letters, respectively. $\tr{\cdot}$, $\fnorm{\cdot}$, and $(\cdot)^\H$ denote respectively the trace, Frobenius norm, and Hermitian transpose. $\mat{I}$ is an identity matrix.

%% file: sections/system_model.tex
\section{System Model} \label{sec:sys_model}
Consider a set of transmitter-receiver pairs $\setK = \br{1,\ldots,K}$ operating in the same spectral band. The transmitters and receiver are equipped with multiple antennas such that transmitter $j$ uses $M_j$ antennas and receiver $k$ uses $N_k$ antennas. The flat fading channel matrix from a transmitter $j$ to a receiver $k$ is $\bH_{jk} \in \mathbb{C}^{N_k \times M_j}$. The received signal at a receiver $k$ is 
\begin{equation} \label{eq:received_signal}
\mat{y}_k = \sum\nolimits_{j=1}^K \sqrt{\gamma_{jk}} \bH_{jk} \bV_j \mat{s}_j + \bz_k,
\end{equation}
where $\mat{s}_j$ is the transmitted signal vector of dimension $d_j$, $\bV_j \in \mathbb{C}^{M_j \times d_j}$ is the transmit precoding matrix at transmitter $j$, $\gamma_{jk}$ is the pathloss coefficient, and $\bz_k \sim \mathcal{CN}(0, \sigma^2 \bI)$ is additive white Gaussian noise. The transmit precoding matrix $\bV_j$ used at transmitter $j$ is restricted to a power constraint such as $\tr{\bV_j \bV_j^\H} \leq P_j, \; j \in \setK.$

Given a pre-log factor $\beta \in [0,1]$, determined by the lost temporal degrees of freedom due to training and overhead \cite{Peters2012, Ayach2012a}, the spectral efficiency of link $k$ is given as
\begin{equation} \label{eq:rate}
R_k = \beta \log_2 \abs{\bI + \gamma_{kk} \bZ_k^{-1} \bH_{kk} \bV_k \bV_k^\H \bH_{kk}^\H},
\end{equation}
where $\bZ_k = \sigma^2 \bI + \sum\nolimits_{j\neq k} \gamma_{jk} \bH_{jk} \bV_j \bV_j^\H \bH_{jk}^\H$ is the interference and noise covariance matrix.

\subsection{CSI-T Sharing Set}
The precoding matrices are designed based on the current network conditions in the form of CSI. We consider an FDD system, meaning that CSI at the receivers (CSI-R) can be obtained by channel training, but due to the non-reciprocal nature of the channel, feedback must be used in order to obtain CSI at the transmitters (CSI-T). Both the channel training and the CSI feedback lead to overhead, whose impact on the spectral efficiency depends on the channel coherence time.

In our model, each receiver $k$ has a non-zero channel from each transmitter $j$ (cf. \eqref{eq:received_signal}). The corresponding links are indexed by the index set $\setJ = \br{(j,k) \mid (j,k) \in \setK \times \setK}$. Due to the path loss factor $\gamma_{jk}$ in \eqref{eq:received_signal}, it is clear that some of the cross-links will interfere more strongly than others. We thus define $\setI \subseteq \setJ$, which will be called a \emph{CSI-T sharing set}. This set specifies what CSI is fed back from the receivers to the transmitters, and should thus correlate with which cross-links are deemed important to treat in the precoder design. Clearly, $(k,k) \in \setI, \; \forall \, k \in \setK$. The output of the stable matching algorithm in Sec.~\ref{sec:matching} will be such a CSI-T sharing set.

The CSI-T sharing set $\setI$ describes how \emph{incomplete} the CSI-T is. This can range from minimal CSI-T ($\setI = \{ (k,k) \mid k \in \setK \}$), to complete CSI-T ($\setI = \setJ$). In certain scenarios, it has been shown that the feasibility of interference alignment can be retained under some level of incomplete CSI-T \cite{Kerret2014}. Given a $\setI \subseteq \setJ$, we now detail how the CSI acquisition will take place, and the corresponding overhead.

\subsection{Overhead Model}
The channel estimation is based on pilot-assisted training and the feedback is based on analog feedback \cite{Ayach2012a}. We assume a block fading model, where the channel is constant for $T$ symbol intervals (a \emph{coherence block}). The CSI must be estimated once per coherence block.

We now propose a simple overhead model. For the channel training, we assume that minimal training is sufficient, i.e. that each channel coefficient can be identified using a single pilot symbol. For the analog feedback, we assume that the minimal number of symbol intervals for orthogonalizing the feedback between users is sufficient for acceptable performance. Admittedly, these assumptions might be coarse approximations of proper system design at low signal to noise ratio (SNR), but they allow us to clearly compare the overhead of training and CSI feedback. Similar approximations have been used before, e.g. in \cite{Peters2012}. Given these assumptions, we now count the number of symbol intervals that are needed in the CSI acquisition to get a measure of the overhead.

\subsubsection{Phases of CSI Acquisition}
The CSI acquisition has the following phases, which are repeated in each coherence block.
\begin{itemize}
    \item[T1] \textbf{Downlink and uplink channel training.} In this phase, the receivers acquire CSI-R which will be fed back to the transmitters during the \emph{F} phase, as well as being used for formulating the receiver utilities in the \emph{SM} phase. To enable succesful decoding in the \emph{F} phase, the uplink channels are also estimated. For pilot orthogonality reasons, this phase requires $\sum_{k=1}^K M_k + N_k$ symbol intervals.
    \item[SM] \textbf{Stable matching.} This phase determines the CSI-T sharing set $\setI$ using the stable matching algorithm in Sec.~\ref{sec:matching}. This phase requires $L_\text{SM}$ symbol intervals of communication, which will be quantified in Sec.~\ref{sec:matching}.
    \item[F]\textbf{Analog CSI feedback.} In this phase, the transmitters acquire the CSI-T needed for the precoder design. Given a CSI-T sharing set $\setI \subseteq \setJ$, receiver $k$ feeds back the CSI for $j \in \{i \mid (i,k) \in \setI \}$ to all transmitters $i \in \setK$. Assuming analog feedback, distributed processing \cite[Sec. III.A]{Ayach2012a} and $N_k \leq M_j$, for all $k, j \in \setK$, this phase requires $\sum_{(j,k) \in \setI} M_j$ symbol intervals. For further details of the feedback mechanics and the resulting symbol interval overhead, see \cite[Sec. III.A]{Ayach2012a}.
    \item[T2] \textbf{Downlink effective channel training.} After precoder optimization, the resulting \emph{effective channels} (i.e. channel matrices multiplied by precoders) are estimated. This requires $\sum_{k=1}^K d_k$ symbol intervals.
\end{itemize}
Summing up the number of symbol intervals needed for training and feedback, the CSI acquisition overhead is
\begin{equation}\label{eq:overhead}
 L_\text{CSI} = {\sum\nolimits_{k=1}^K \left( M_k + N_k + d_k \right)} + {\sum\nolimits_{(j,k) \in \setI} M_j}.
\end{equation}
Note that the first term in the summation (accumulated training overhead) is linear in $K$, whereas the second term (feedback overhead) in the worst case becomes quadratic in $K$. Hence, in terms of overhead reduction, there are large gains to be anticipated by reducing the amount of feedback.

After the five phases of CSI acquisition, data transmission takes place during the remaining $L_\text{data} = T - L_\text{CSI} - L_\text{SM}$ symbol intervals. Then, the pre-log factor in \eqref{eq:rate} is 
\begin{equation}\label{eq:prelog}
\beta = {L_\text{data}}/{T} = 1 - ({L_\text{CSI} + L_\text{SM}})/{T}.
\end{equation}
It is now clear that optimizing $R_k$ in \eqref{eq:rate} becomes a tradeoff between better interference management (higher spectral efficiency factor) and lower overhead (higher pre-log factor). Alhough the above method is inherently distributed, it enables transmitters within the CSI-T sharing set to have ``global CSI-T'' and thus perform the precoding in a centralized fashion. Note however that the complexity of exhaustively searching all possible CSI-T sharing sets is $2^{K^2-K}$. Next, we detail an efficient and distributed procedure based on stable matching which will determine the CSI-T sharing set.

%% file: sections/stable_matching.tex

\section{Distributed CSI Selection}\label{sec:matching}
Many-to-many stable matching has been of interest for its application in the job matching problem \cite{Roth1984}. There, a set of firms and a set of workers exist, where each firm has a set of vacant positions to offer to workers, and each worker can work at more than one firm. The interest of a firm is to hire the best workers and each worker's interest is to work at the most preferred combination of firms. The solution of the job matching problem is a many-to-many stable matching. 

The job matching problem relates to the problem in our setting in the following. First, we seek a matching between the set of transmitters (firms) and the receivers (workers) which dictates the CSI-T sharing set. Second, the stability of the solution concept supports a distributed implementation of the mechanism.



\subsection{Stable Matching Model}
Consider the set of transmitters $\setT = \br{tx_1,\ldots,tx_K}$ and the set of receivers $\setR= \br{rx_1,\ldots,rx_K}$. A matching between the two sets is defined as follows.%
\begin{definition}\label{def:Matching}
A matching $\setM$ is a correspondence from the set $\setT \cup \setR$ to the set of all subsets of $\setT \cup \setR$ and satisfies the following properties for $j \in \setT$ and $k \in \setR$:
\begin{enumerate}
\item $\setM(j) \in \setR \cup \emptyset$, and $\setM(k) \in \setT \cup \emptyset$,
\item $k \in \setM(j) \text{ if and only if } j \in \setM(k)$.
\end{enumerate}
\end{definition}

The matching $\setM$ is a set valued function such that $\setM(j)$ is the set of receivers matched to transmitter $j\in \setT$ and $\setM(k)$ is the set of transmitters matched to receiver $k\in \setR$. If a transmitter $j$ is unmatched, then $\setM(j) = \emptyset$. Similarly, $\setM(k) = \emptyset$ means that receiver $k$ is unmatched. Condition 2. in Definition~\ref{def:Matching} ensures that whenever a transmitter $j$ is matched to a receiver $k$ then $k$ would be also matched to transmitter $j$.

\subsubsection{Receiver and Transmitter Preference Sets}
Each receiver $k \in \setR$ must have a strict, transitive and complete preference relation over the set $2^\setT$ containing all subsets of $\setT$. For a given set of transmitters $\setB \subseteq \setT$, a receiver $k$ is able to select the most preferred subset of $\setB$ by solving:
\begin{equation}\label{eq:pref_profile_receiver}
\mathrm{C}^{rx}_k(\setB) = \argmax_{\setS \in \setB} ~~ \sum\nolimits_{j \in \setS} \phi^{rx}_k(j) ~~ s.t. ~~ \abs{\setS} \leq q_k^\mathrm{rx},
\end{equation}
\noindent where the functions $\phi^{rx}_k(j)$ are defined as
\begin{equation}\label{eq:phi_rx}
\phi^{rx}_k(j) =
\begin{cases} \frac{\sqrt{\gamma_{jk}} \fnorm{\bH_{jk}^\H \bH_{kk} }}{\sqrt{\gamma_{kk}} \sfnorm{\bH_{kk}}} &\mbox{if } j\in \setT \\ 
0 & \mbox{if } j = \emptyset
\end{cases},
\end{equation}
and $q_k^\mathrm{rx} \in \mathbb{N}$ is called the matching quota of receiver $k$. The function in \eqref{eq:phi_rx} reflects a measure on how much the channel from transmitter $j$ is aligned to the direct channel. This model shares similarities with the utility functions formulated in \cite{Jorswieck2013} in the context of cognitive radio. If this measure is large, then the corresponding transmitter $j$ could potentially generate substantial amount of interference at receiver $k$, and hence an effort in terms of CSI feedback and precoding should be taken to manage it. An unmatched receiver has zero utility.


In \eqref{eq:pref_profile_receiver}, a receiver $k$ prefers to be matched with the transmitters which maximize the sum of the measures defined in \eqref{eq:phi_rx}, with the constraint that the total number of matched transmitters is not more than a design integer $q^\mathrm{rx}_k \in \mathbb{N}$. Observe, that Problem \eqref{eq:pref_profile_receiver} can be solved with low complexity by a greedy method since the objective functions are additively separable \cite{Federgruen1986}.

Similarly, each transmitter $j$ must have a strict, transitive and complete preference relation over the set $2^\setR$ of all subsets of $\setR$. Given a set of receivers $\setB \subseteq \setR$, we define the subset of $\setB$ which transmitter $j$ prefers the most as
\begin{equation}\label{eq:pref_profile_transmitter}
\mathrm{C}^{tx}_j(\setB) = \argmax_{\setS \in \setB} ~~ \sum\nolimits_{k \in \setS} \phi^{tx}_j(k) ~~ s.t. ~~ \abs{\setS} \leq q^\mathrm{tx}_j,
\end{equation}
\noindent where 
\begin{equation}\label{eq:phi_tx}
\phi^{tx}_j(k)
\begin{cases} \gamma_{jk} P_j \sfnorm{\bH_{jk}}/(M_jN_j) &\mbox{if } k\in \setR\\ 
0 & \mbox{if } k = \emptyset
\end{cases},
\end{equation}
and $q^{tx}_j\in \mathbb{N}$ is the quota of transmitter $j$. The function in \eqref{eq:phi_tx} for a transmitter $j$ is increasing in the strength of the channel to a receiver $k$ and decreasing with the number of antennas. Since the number of antennas is proportional to the channel feedback overhead for CSI-T acquisition, the ratio in \eqref{eq:phi_tx} reflects the amount of interference relative to overhead. If this measure is high, then a matching between transmitter $j$ and receiver $k$ is desirable.%
%
%
\subsubsection{Many-to-many Stable Matching}
The first of two requirements for stability is the following.
\begin{definition}\label{def:ind_rational}
Matching $\setM$ is {individually rational} if 
\begin{enumerate}
\item no transmitter $j \in \setT$ exists with $\setM(j) \neq \mathcal{C}^{tx}_j(\setM(j))$, 
\item no receiver $k \in \setR$ exists with $\setM(k) \neq \mathcal{C}^{rx}_k(\setM(k))$.
\end{enumerate}
\end{definition}
Condition 1. says that a matching $\setM$ is not individually rational for a transmitter $j$ if the set of receivers $\setM(j)$ matched to transmitter $j$ are not all within the solution of the optimization in \eqref{eq:pref_profile_transmitter} with $\setM(j)$ as input. Analogously condition 2. for a receiver $k$. 

\begin{definition}\label{def:blocked}
Matching $\setM$ is pairwise stable if there does not exist a pair $(k,j) \in \setR \times \setT$ such that 
\begin{enumerate}
\item $k \neq \setM(j)$
\item $k \in \mathcal{C}^{tx}_j(\setM(j) \cup \{k\})$ and $j \in \mathcal{C}^{rx}_k(\setM(k) \cup \{j\})$
\end{enumerate}
\end{definition}
Pairwise stability requires that there exist no receiver $k$ and no transmitter $j$ which are not matched to each other but prefer a matching between themselves. The two conditions in the second requirement mean, respectively, that transmitter $j$ is in the solution set of receiver $k$'s optimization problem in \eqref{eq:pref_profile_receiver} given the set $\setM(k) \cup \{j\}$ as input, and receiver $k$ would be in the solution set of the optimization problem of transmitter $j$ in \eqref{eq:pref_profile_transmitter} with $\setM(j) \cup \{k\}$ as input.

\begin{definition}\label{def:stable_matching}
A matching $\setM$ is \emph{stable} if it is individually rational (Definition \ref{def:ind_rational}) and pairwise stable (Definition \ref{def:blocked}).
\end{definition}

Generally, a many-to-many stable matching need not exist. We prove this issue positively in the following.
\begin{theorem}
A stable matching exists in our setting.
\end{theorem}
\begin{proof}
A sufficient condition for the existence of a stable matching (Definition \ref{def:stable_matching}) is when the preferences of the transmitters and receivers satisfy the gross substitute property \cite{Roth1984}. Our formulation of the utility functions satisfy this property according to \cite[Section 2]{Gul1999} since \eqref{eq:pref_profile_receiver} and \eqref{eq:pref_profile_transmitter} correspond to the q-satiation of the performance. 
\end{proof}
A stable matching in our setting depends strongly on the choice of quotas for each transmitter and receiver, incorporated in problems \eqref{eq:pref_profile_receiver} and \eqref{eq:pref_profile_transmitter}. Efficient modeling for the quotas seems to be a hard task which will depend on several parameters in the considered scenario. For a symmetric system, it is possible to determine the quotas as a function of required pre-log factor.
\begin{theorem}\label{thm:quotas}
Assume a symmetric network in which $M_k = M, N_j = N, d_k = d$ for all $k \in \setK$. In addition let $q^{rx}_k = q^{tx}_k = q$ for all $k \in \setK$, then for a given minimum desired pre-log factor $\hat{\beta} \in [{\beta}^{min},\beta^{max}]$ where 
\begin{align} \label{eq:betamin}
{\beta}^{min} & = 1 - {K(2M + N + d)}/{T},\\ \label{eq:betamax}
\beta^{max} &= 1 - {K(KM + M + N + d)}/{T},
\end{align}
the following condition should hold
\begin{equation}\label{eq:cond_quota}
q \leq \lfloor{T(1-\hat{\beta})}/{KM} - 1 - {(N+d)}/{M}\rfloor.
\end{equation}
\end{theorem}
\begin{proof}
First observe that since each transmitter and receiver has the same quota $q$, each transmitter (receiver) will be matched with $q$ other receivers (transmitters) in a stable matching. This fact follows from the individual rationality assumption (Definition \ref{def:ind_rational}). Consequently, the analog CSI feedback overhead is equal to $qKM$. Using \eqref{eq:prelog} we can formulate the condition on the quota $q$ in \eqref{eq:cond_quota}. The bounds on $\hat{\beta}$ correspond to minimal and full CSI-T sharing sets.
\end{proof}

\subsubsection{Algorithm and Distributed Implementation}
Algorithm~\ref{alg:SM} describes the steps to reach a many-to-many stable matching which is based on the proposed algorithm in \cite{Roth1984}. In Step 3, each receiver ``applies'' to its most preferred transmitters by sending a message to each of them. This message includes the norm of the interference channel from the transmitter to the applying receiver which is needed for the transmitter's optimization problem in \eqref{eq:pref_profile_transmitter}. The overhead for sending this message is assumed to be a single symbol interval (Step 4). On receiving all applications, each transmitter accepts the receivers based on its optimization problem in \eqref{eq:pref_profile_transmitter} and rejects the rest. Here, the binary decisions are communicated to the applying receivers using a single symbol interval. 

During the algorithm, the receivers apply only to the transmitters which have not rejected them before. The transmitters, upon obtaining the applications from the receivers, choose the ones they prefer most and which they have not rejected before. Algorithm \ref{alg:SM} is guaranteed to converge to a stable matching \cite[Proposition 5]{Roth1984}. Moreover, the algorithm's complexity is fairly low given that each receiver applies at most one time to any transmitter \cite[Proposition 3]{Roth1984}.
\begin{algorithm}[t!]
\begin{algorithmic}[1]
\Statex Initialize:  $i = 0, \setT^0_k = \setT$, for all $k \in \setR$, $L_{SM} = 0$
\Repeat
\For{$k \in \setR$} 
\State \parbox[t]{\dimexpr\linewidth-\algorithmicindent-\algorithmicindent}{Receiver $k$ applies to $\mathcal{C}^{rx}_k(\setT^i_k)$ with no transmitter in $\setT^i_k$ has rejected him before.\strut}
\State \parbox[t]{\dimexpr\linewidth-\algorithmicindent-\algorithmicindent}{$L_{SM} = L_{SM} + 1$;\strut}
\EndFor
\For{$j \in \setT$} 
\State \parbox[t]{\dimexpr\linewidth-\algorithmicindent-\algorithmicindent}{Transmitter $j$ accepts $\mathcal{C}^{tx}_k(\setP^i_j)$ with no receiver in 
\begin{equation}\label{eq:applications_to_j}
\setP^i_j = \{\setS \subseteq \setR \mid j \in \mathcal{C}^{rx}_k(\setT^i_k),\; \forall k \in \setS\}
\end{equation}
it rejected before.\strut}
\State Transmitter $j$ rejects $\setP^i_j \setminus \mathcal{C}^{tx}_k(\setP^i_j).$
\State \parbox[t]{\dimexpr\linewidth-\algorithmicindent-\algorithmicindent}{$L_{SM} = L_{SM} + 1$;\strut}
\EndFor
\State $i = i + 1$;
\Until{no receiver applies to any transmitter}
\end{algorithmic}
\caption{\label{alg:SM}Stable matching algorithm \cite{Roth1984}.} 
\end{algorithm}%
%
%

%% file: sections/simulations.tex

\section{Numerical Results}\label{sec:simulations}
We consider $25$ transmitter-receiver pairs. The transmitters are uniformly distributed in a $250 \times 250$ m$^2$ region, and each receiver is randomly located at a distance of $50$ m from its transmitter. We set the number of antennas as $N_k = M_j = 5$ for all $j,k$, and the number of data streams per user is two. We assume Rayleigh fading such that $\mat{H}_{jk} \sim \mathcal{CN}(0,\mat{I})$ for all $j,k \in \setK$. Let the distance between transmitter $j$ and receiver $k$ be $a_{jk}$. The pathloss coefficient is modeled as $\gamma_{jk} = a_{jk}^{-3}$. The coherence block is assumed $T = 10^4$ symbol intervals. 

Utlizing Theorem~\ref{thm:quotas} to determine the quotas, we let the desired pre-log factor to be decreasing in $\mathrm{SNR}:=({a_{kk}^{-\delta}})/{\sigma^{2}}$ as $\hat{\beta} = \alpha(SNR) \beta^{min} + (1-\alpha(SNR))\beta^{max}$, where we introduce the function $\alpha(SNR) = \log(1 + SNR/(1+SNR))$.

In \figurename~\ref{fig:sum_rate}, the average sum spectral efficiency in \eqref{eq:rate} is plotted for different CSI selection schemes. The average is taken over $250$ random system deployments. The overhead in CSI training corresponding to the plots in \figurename~\ref{fig:sum_rate} is shown in \figurename~\ref{fig:time}. In the minimal CSI-T sharing scheme, the transmitters only have CSI-T regarding the direct channels. Precoding matrices are then calculated according to the eigenvectors of the direct channels. In the full CSI-T sharing scheme, complete CSI-T is present at all transmitters. This scheme requires the largest amount of overhead in the F phase (Sec. 2.2.1) as is shown in \figurename~\ref{fig:time}. For calculating the transceivers, we use the alternating optimization algorithm in \cite{Shi2011} in which we fix the maximum number of iterations to five. It can be seen in \figurename~\ref{fig:sum_rate} that the performance of minimal CSI-T sharing is relatively high at low SNR while at high SNR interference management is necessary despite a high overhead in CSI-T acquisition. 

The stable matching CSI-T selection scheme corresponds to Algorithm~\ref{alg:SM}. After the F phase, the algorithm from \cite{Shi2011} is applied at the transmitters using the possibly incomplete CSI-T. The performance of the stable matching scheme is shown to outperform minimal and full CSI-T sharing. By neglecting the induced overhead, we compare our distributed algorithm with the centralized greedy selection scheme proposed in \cite{Westreicher2012} for determining the users which perform interference alignment in partially connected networks. It can be observed that the performances of the two algorithms are very similar, where we have used our novel model for the quotas for both algorithms. When consider the actual overhead $L_{SM}$, it can be seen that the incurred performance loss is very small in stable matching. This overhead $L_{SM}$, which is influenced by the complexity of the stable matching algorithm, is relatively small compared to the overhead in CSI feedback.

\begin{figure}[t]
  \centering
  \includegraphics[width=\linewidth,clip]{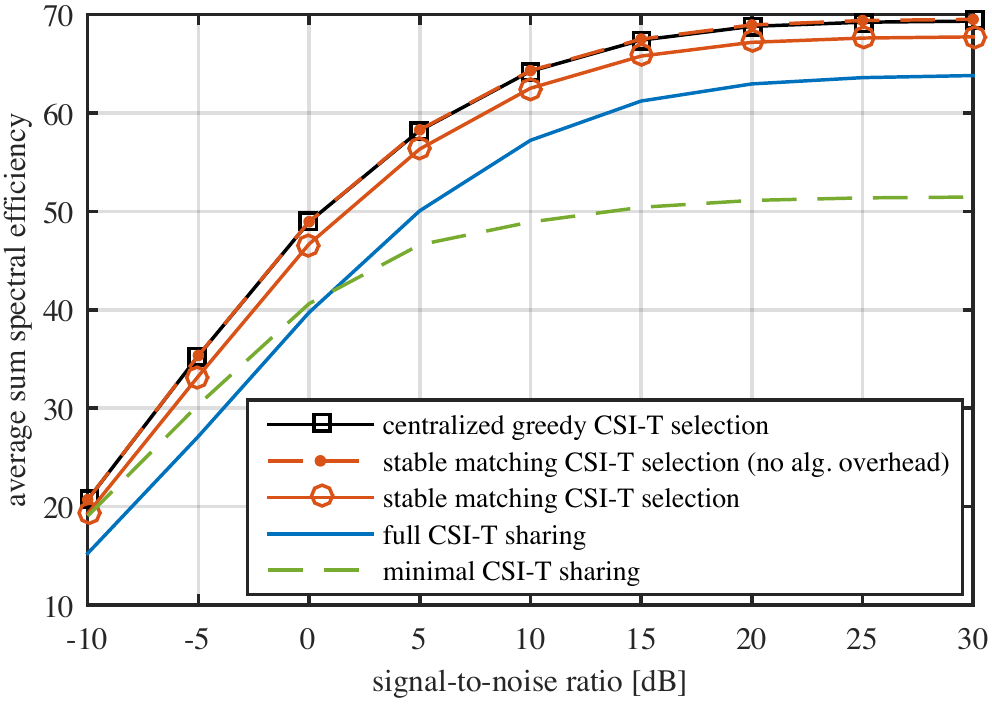}
  \vspace{-.7cm}
  \caption{\label{fig:sum_rate} Performance of different CSI-T selection schemes.}
\end{figure}
\begin{figure}[t]
  \centering
  \includegraphics[width=\linewidth,clip]{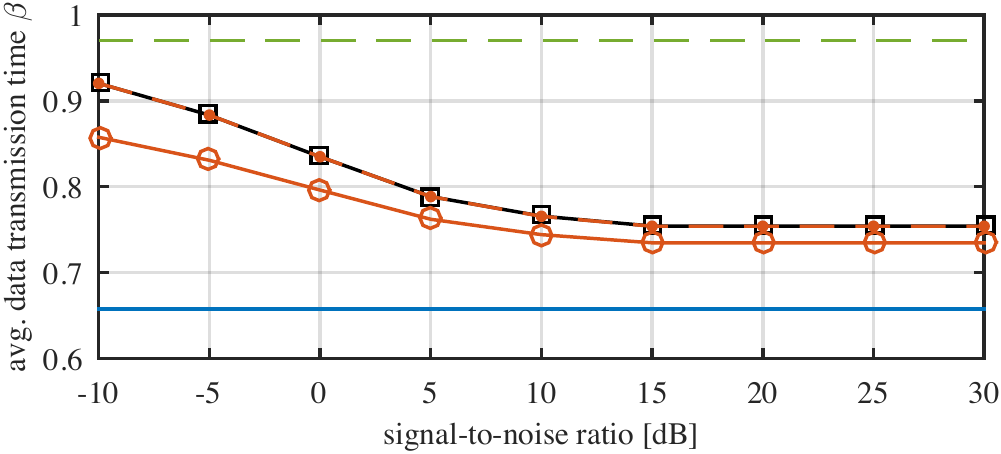}
  \vspace{-.7cm}
  \caption{\label{fig:time} Corresponding average data transmission time $\beta$.}
\end{figure}